\documentclass[11pt]{article}
\usepackage{geometry}                
\usepackage{graphicx}
\usepackage{amssymb}
\usepackage{epstopdf}
\usepackage{amsmath}
\usepackage{amsthm}
\usepackage[utf8]{inputenc}
\usepackage{amsmath}
\usepackage{tensor}
\newcommand{\half}{\frac{1}{2}}
\newcommand{\J}{{\cal J} }
\newcommand{\C}{{\cal C} }
\newtheorem{theorem}{Theorem}
\newtheorem*{theorem*}{Theorem}
\newtheorem*{corollary*}{Corollary}
\newtheorem{definition}{Definition}
\newtheorem{proposition}{Proposition}
\newtheorem*{corollary}{Corollary}

\numberwithin{equation}{section}

\title{\bf Affine Symmetry, Geodesics, and Homogeneous Spacetimes}
\author{David Maughan and Charles Torre\\{\it Department of Physics, Utah State University, Logan, Utah, USA}}

\date{July 21, 2018}                                           

\begin{document}

\maketitle

\abstract{We show that the conservation laws for the geodesic equation which are associated to affine symmetries can be obtained from  symmetries of the  Lagrangian for affinely parametrized geodesics according to Noether's theorem, in contrast to claims found in the literature.  In particular, using Aminova's classification of affine motions of Lorentzian manifolds, we show in detail how affine motions define generalized  symmetries of the geodesic Lagrangian. We compute all infinitesimal proper affine symmetries and the corresponding geodesic conservation laws for all homogeneous solutions to the Einstein field equations in four spacetime dimensions with each of the following energy-momentum contents: vacuum,  cosmological constant, perfect fluid, pure radiation, and homogeneous electromagnetic fields. }

\section{Introduction}

Homotheties of a  metric define point symmetries of the Lagrangian for geodesics and define conservation laws for the geodesic equation via Noether's theorem \cite{Katzin1}, \cite {Katzin}, \cite{Prince1}, \cite{Prince2},  \cite{Tsamparlis}.   In most cases the analytic tractability of a system of geodesic equations depends upon the existence of such conservation laws, or perhaps  conservation laws associated with other geometric structures, {\it e.g.,} Killing tensors. {\it Affine symmetries} are diffeomorphisms of a spacetime which preserve the affine connection (see, {\it e.g.}, \cite{Eisenhart}).   
These include  homotheties (and isometries)  as a special case, but the group of affine transformations may include non-homothetic transformations.  Affine symmetries which are not homotheties are called {\it proper affine symmetries}.  Affine symmetries act as a transformation group on the space of solutions of the affinely parametrized geodesic equation \cite{Katzin1}, \cite{Prince1}.  For this reason they are often called {\it affine collineations}.

It has been known for some time that there are two conservation laws for the geodesic equation which are associated to each 1-parameter group of proper affine symmetries \cite{Katzin}, \cite{Hojman1}.  This may be  surprising  since proper affine symmetries do not define point transformations which preserve the geodesic Lagrangian. Indeed these conservation laws have been characterized as ``non-Noetherian'' in \cite{Hojman1}, \cite{Hojman2}.  To some extent, the existence of these conservation laws has been understood  in the context of modifications of the Lagrangian formalism such as found in \cite{Prince2} and in \cite{Hojman1}.  As we shall show here, one can directly apply Noether's theorem  to the standard  Lagrangian for affinely parametrized geoedesics to obtain  the two conservation laws associated to proper affine symmetries. To do this we  use the fact that to account for all conservation laws of a system of Euler-Lagrange equation one must account for all {\it generalized symmetries} of the Lagrangian \cite{Olver}, \cite{Prince1}, \cite{Prince2}.\footnote{See reference \cite{Olver} for a comprehensive exposition of generalized symmetries and Noether's theorem, applicable to PDEs and ODEs.  See reference \cite{Prince1}, \cite{Prince2} for a geometric exposition of the theory of symmetries and conservation laws tailored to second order ODEs with applications to projective symmetries and conservation laws of the geodesic equation as well as to conservation laws associated to Killing tensors. }  Generalized symmetries need not act as point transformations, but instead act as infinitesimal transformations on the infinite jet space of the dependent variables. These symmetries were introduced by Noether \cite{Noether}; they generalize  point symmetries and  contact symmetries.  We shall show that associated to each proper affine motion there  are two generalized symmetries of the geodesic Lagrangian.  Noether's theorem yields the corresponding conservation laws.  Using Aminova's classification of affine symmetries \cite{Aminova} we explain in some detail how the infinitesimal transformations associated to affine symmetries manage to define generalized symmetries of the geodesic Lagrangian.  
The derivation of the geodesic conservation laws for affine symmetries from Noether's theorem is the principal result of our paper. 

A number of papers have found affine symmetries  for various solutions of the Einstein equations, {\it e.g.,} \cite{Beldran}, \cite{Maartens}, \cite{Collinson}, \cite{Hojman2}. Hall and da Costa \cite{Hall} have given a classification (based upon holonomy groups) of possible affine symmetries which can occur in four-dimensional spacetimes.   The possibilities for electrovacua in four dimensions have been examined in \cite{Carot}.  As  a modest contribution to this body of work and as an illustration of our results, we calculate all continuous proper affine symmetries and corresponding conservation laws for all homogeneous solutions to the Einstein-matter field equations in four dimensions for each of  the following energy-momentum contents: vacuum,  cosmological constant,  perfect fluid, pure radiation,  and homogeneous electromagnetic field.   To our knowledge the proper affine symmetries have not been exhaustively enumerated for all such solutions. 

In \S\ref{affinesection} we will review the fundamentals of affine symmetry and briefly review the results of Aminova's classification of continuous affine symmetries of Lorentzian manifolds in any dimension.   In \S\ref{geosection} we summarize the results we will need from the theory of generalized symmetries and conservation laws in the context of ordinary differential equations. We then show how affine symmetries define generalized  symmetries of the geodesic Lagrangian.  We apply Noether's theorem to obtain the corresponding conservation laws. Finally, in \S\ref{solsection} we enumerate all homogeneous solutions of the Einstein equations (in four spacetime dimensions and with the matter content as listed above) along with all their infinitesimal proper affine symmetries and corresponding conservation laws.

\section{Affine Symmetry}
\label{affinesection}
In this section we review the fundamentals of affine symmetry transformations, with an eye on the applications to the geodesic equation and solutions to the Einstein equations given in the following sections. 

Let $(M, g)$ be a pseudo-Riemannian manifold.  The metric uniquely determines a torsion-free affine connection $\nabla$ from the condition $\nabla g = 0$.  A diffeomorphism $\phi\colon M\to M$ is an {\it affine symmetry} if it preserves this connection, that is, for any tensor field $T$
\begin{equation}
\phi^*\left(\nabla T\right) =   \nabla\left(\phi^*T\right).
\label{affinedef}
\end{equation}
It is easy to check that homotheties
($\phi^* g = c\, g$, $c=const.$) are affine symmetries.  Affine symmetries which are not homotheties will be called {\it proper affine symmetries}.  The existence of a proper affine symmetry implies  the vector space of parallel  symmetric $\left(0\atop2\right)$ tensor fields has dimension greater than one since  $h = \phi^* g$ is parallel:
\begin{equation}
\nabla h = \nabla (\phi^*g) = \phi^*(\nabla g) = 0.
\end{equation}

A 1-parameter group $\phi_\lambda$ of diffeomorphisms is an {\it affine motion} if it preserves the connection for each value of the parameter $\lambda$.  In this case the definition (\ref{affinedef}) can be replaced with an infinitesimal condition involving the Lie derivative along  the {\it affine vector field}  $Y$ on $M$ generating the 1-parameter group:
\begin{equation}
L_Y (\nabla T) = \nabla (L_YT),\quad \forall\ T.
\label{symm}
\end{equation}
This condition is equivalent to
\begin{equation}
\nabla_a \nabla_b Y^c = R^c{}_{bad} Y^d,
\label{infAffine}
\end{equation}
where $R^a{}_{bcd}$ is the Riemann tensor and we are using the abstract index notation.  The affine vector field therefore satisfies an over-determined system of linear partial differential equations of finite type. For a generic metric $g$ there are no solutions.  The maximum number of solutions is $n(n+1)$ where $n = {\rm dim}(M)$.  Homothetic vector fields,  defined by
\begin{equation}
L_Y g = c\, g, \quad c = {\rm constant},
\end{equation} 
satisfy (\ref{infAffine}). An affine vector field which is not a homothetic vector field will be called a {\it proper affine vector field}.  Proper affine vector fields come in equivalence classes: two proper affine vector fields belong to the same equivalence class if they differ by a homothetic vector field.   In light of (\ref{symm}), a proper affine vector field $Y$ defines a  parallel symmetric $\left(0\atop2\right)$ tensor field $h$ not proportional to the metric via 
\begin{equation}
h = L_Y g, \quad \nabla h = 0.
\end{equation}

To our knowledge the classification of affine motions is not complete except in the cases of Riemannian and Lorentzian manifolds.   All affine motions of an irreducible Riemannian manifold are homotheties \cite{KN}.  In the reducible case, the de Rham theorem decomposes a Riemannian manifold as a product of a flat manifold and irreducible Riemannian manifolds of dimension greater than one:
\begin{equation}
(M, g) = (M_0 \times M_1\times M_2 \times \dots \times M_r, g_0 + g_1 + g_2 + \dots + g_r),
\label{prod}
\end{equation}
where $g_0$ is flat and is the restriction of $g$ to the submanifold tangent to the distribution of parallel vector fields. 
It follows that affine vector fields generate homotheties in each irreducible component.  In the Lorentzian case, Aminova \cite{Aminova} has shown\footnote{See also the results of Hall in four dimensions \cite{Hall}.} that  the preceding result holds, now corresponding to the de Rham-Wu decomposition \cite{Wu}, but a new possibility arises.  A locally irreducible Lorentz manifold $(M, g)$ may admit a proper affine vector field $Y$, but only if it admits a parallel null vector field $k^a$ and the affine vector field acts via
\begin{equation}
L_Y g_{ab} = \alpha g_{ab} + \beta k_a k_b,\quad  \nabla_a k_b = 0, \ k_ak^a=0, \quad \alpha, \beta {\rm \ constant},\ \beta\neq0.
\label{irr}
\end{equation}
In this case there will exist coordinates $(w, v, x^i)$, $i = 1, 2, \dots, n-2$, in which 
\begin{equation}
g = dw \otimes dv + dv \otimes dw + g_{ij}(w, x^i) dx^i \otimes dx^j,
\label{irr2}
\end{equation}
and, modulo the addition of a homothetic vector field,
\begin{equation}
 Y = w\,\partial_v.
 \label{irr3}
 \end{equation}
In the Lorentzian case affine motions occur only when either of these two situations (\ref{prod}), (\ref{irr2}) (or both) arise.  The proper affine vector fields are then of 3 types: (I) homothetic vector fields for the irreducible subspaces, (II) vector fields acting as in (\ref{irr}), or  (III) infinitesimal generators of linear ``intermixing'' transformations among the coordinates adapted to the parallel vector fields.  In case III the affine vector fields can be put into the form 
\begin{equation}
Y = y\partial_v, w\partial_z, z\partial_y, \quad {etc.}
\label{mix}
\end{equation} 
where $w$ and $v$ are  defined in (\ref{irr2}) and $y, z$, {\it etc.}, denote  coordinates on $M_0$ which rectify  non-null parallel vector fields $\partial_y$, $\partial_z$, {\it etc}. 

\section{Affine Symmetries and Conservation Laws of the Geodesic Equation}
\label{geosection}

In this section we will obtain the principal result of this paper: a derivation of the conservation laws associated to affine symmetries from Noether's theorem.  To this end, we begin with some definitions from the geometry of differential equations and the calculus of variations, which have been specialized to the case of one independent variable \cite{Olver} (see also \cite{Prince1}, \cite{Prince2}).  

\subsection{\it Preliminaries}

Let $\C$ be the bundle of curves on $M$. Let $\J$ be the infinite jet bundle of curves in $M$ \cite{Saunders}.   Using a coordinate chart $x^\alpha$ on $U\subset M$, and a parameter $s$ for the curve, local coordinates on $\C$ are $(s, x^\alpha)$.  A curve in $M$ is then a cross section, $x^\alpha = u^\alpha(s)$,  of $\cal C$. Local coordinates on $\J$ are denoted by $(s, x^\alpha, \dot x^\alpha, \ddot x^\alpha, \dots)$.  The cross section $x^\alpha = u^\alpha(s)$ extends to a cross section of $\J$ via
\begin{equation}
x^\alpha = u^\alpha(s),\quad \dot x^\alpha = {du^\alpha\over ds},\quad \ddot x^\alpha = {d^2u^\alpha\over ds^2}, \quad \cdots 
\end{equation} 
Functions on $\J$ are denoted by
\begin{equation}
f[x] \equiv f(s, x, \dot x, \ddot x, \dots).
\end{equation}
The {\it total derivative}, $D\colon C^\infty(\J) \to C^\infty(\J)$, is defined by
\begin{equation}
D f[x] = {\partial f\over \partial s} + {\partial f\over \partial x^\alpha} \dot x^\alpha + {\partial f\over\partial \dot x^\alpha} \ddot x^\alpha + \dots,
\end{equation}
and represents the ``total time derivative''  along a curve $x^\alpha = u^\alpha(s)$ in the sense that
\begin{equation}
\left(Df[x]\right)_{x = u(s)} = {d\over ds} \left(f[u(s)]\right).
\end{equation}

The tangent space at a given point in $\J$ is spanned by
\begin{equation}
\left\{{\partial \over \partial s},\quad {\partial\over\partial x^\alpha},\quad {\partial\over\partial \dot x^\alpha},\quad {\partial\over\partial \ddot x^\alpha},\quad  \dots\right\}
\end{equation}
A vector field ${\bf v}$ on $\J$ is called a {\it generalized vector} field if it takes the form (in a local coordinate chart)
\begin{equation}
{\bf v} = a[x] {\partial\over \partial s} + b^\alpha[x] {\partial\over \partial x^\alpha}.
\label{gvf}
\end{equation}
In the calculus of variations, generalized vector fields correspond to infinitesimal variations of curves $x^\alpha = u^\alpha(s)$ via
\begin{equation}
\delta u^\alpha(s)  = b^\alpha[u(s)] - {d u^\alpha(s)\over ds} a[u(s)].
\label{deltaudef}
\end{equation}
Given a generalized vector field {\bf v}, its {\it infinite prolongation} ${\rm pr}\, {\bf v}$ is the extension to $\cal J$ given by
\begin{equation}
{\rm pr}\, {\bf v} = {\bf v} + c_1^\alpha[x] {\partial\over \partial \dot x^\alpha} + c_2^\alpha[x] {\partial\over \partial \ddot x^\alpha}  + \dots,
\end{equation}
where
\begin{equation}
c_1^\alpha = D \left(b^\alpha[x] - a[x] \dot x^\alpha\right) + a[x] \ddot x^\alpha, \quad c_2^\alpha = D^2 \left(b^\alpha[x] - a[x] \dot x^\alpha\right) + a[x]\,  \dddot{x}^\alpha, \quad \dots
\label{prolong}
\end{equation}
The prolongation of $\bf v$ describes the extension of the variation (\ref{deltaudef}) to all derivatives of  the curve, {\it e.g.,}
\begin{equation}
\delta \left({d^2 u^\alpha(s)\over ds^2}\right)  = c_2^\alpha[u(s)] -   {d^3 u^\alpha(s)\over ds^3} a[u(s)] =  D^2 \left(b^\alpha[x] - a[x] \dot x^\alpha\right)\Big|_{x = u(s)}.
\label{deltaddot}
\end{equation}
From equation (\ref{prolong}) it is clear that, in general, the infinitesimal transformation of a quantity involving $n$ derivatives of the curve will involve derivatives of the curve of order greater than $n$.  Consequently, the infinitesimal transformation defined by a generalized  vector field requires the entire jet bundle for its definition.   The corresponding transformation group on the set of curves -- cross sections of $\C$ -- is constructed by solving an auxiliary system of  PDEs \cite{Olver}.   
A restricted class of infinitesimal transformations is generated by vector fields which can be  defined entirely on the bundle of curves $\C$ and generate a transformation group of $\cal C$.  These are the {\it point transformations}, which arise when the components of $\bf v$ only depend upon $(s, x^\alpha)$:
\begin{equation}
a[x] = a(s, x),\quad b^\alpha[x] = b^\alpha(s, x).
\end{equation}

A generalized vector field (\ref{gvf}) with $a[x] = 0$ is called an {\it evolutionary vector field}.
The prolongation of an evolutionary vector field ${\bf v}_{\rm ev} = \sigma[x] {\partial\over \partial x^\alpha}$ takes the simple form:
\begin{equation}
{\rm pr}\, {\bf v}_{\rm ev} = \sigma[x] {\partial\over \partial x^\alpha} + (D\sigma[x]) {\partial\over \partial \dot x^\alpha} + (D^2\sigma[x]) {\partial \over \partial \ddot x^\alpha} + \cdots.
\end{equation}
   In general, an evolutionary vector field ${\bf v} = \sigma^\alpha[x] {\partial\over\partial x^\alpha}$ defines an infinitesimal variation of a curve $x^\alpha = u^\alpha(s)$ according to:
\begin{equation}
\delta u^\alpha(s) = \sigma^\alpha[u(s)] .
\end{equation}

The total derivative is associated to a generalized vector field ${\partial\over \partial s} + \dot x^\alpha {\partial\over \partial x^\alpha}$ whose infinite prolongation is:
\begin{equation}
D = {\partial\over \partial s} + \dot x^\alpha {\partial\over \partial x^\alpha} +  \ddot x^\alpha {\partial\over \partial \dot x^\alpha} + \dddot x^\alpha {\partial\over \partial \ddot x^\alpha} + \cdots 
\end{equation}
  The prolongation of a generalized vector field (\ref{gvf}) can always be decomposed into the sum of the prolongation of an evolutionary vector field and a total derivative:
\begin{equation}
{\rm pr}\, {\bf v} = a[x] D + {\rm pr}\, {\bf v}_{\rm ev}, 
\end{equation}
where
\begin{equation}
 {\bf v}_{\rm ev} =(b^\alpha[x] - a[x] \dot x^\alpha) {\partial\over\partial x^\alpha} ,
 \end{equation}
is called the {\it evolutionary representative} of $\bf v$.  The evolutionary representative of a vector field generating a point symmetry is of the form
\begin{equation}
{\bf v}_{\rm ev} = \left[b^\alpha(s, x) - a(s, x) \dot x^\alpha\right] {\partial\over\partial x^\alpha}.
\end{equation}
for some functions $a(s, x)$ and $b^\alpha(s, x)$ on $\cal C$.

Let $x^\alpha = u^\alpha(s)$ be a curve in $U$  parametrized by $s$.   
This curve is an affinely parametrized geodesic if and only if it satisfies
\begin{equation}
{d^2 u^\alpha\over ds^2} + \Gamma^\alpha_{\beta\gamma}(u) {du^\beta\over ds} {d u^\gamma\over ds} = 0,
\label{geo}
\end{equation}
where $\Gamma^\alpha_{\beta\gamma}(u)$ are the Christoffel symbols of the metric-compatible connection evaluated along the curve.    The equations (\ref{geo}) are equivalent to the Euler-Lagrange equations of the Lagrangian $L\colon \J\to R$ for affinely parametrized geodesics:
\begin{equation}
L =  - \frac{1}{2} g(\dot x, \dot x) \equiv - \frac{1}{2} g_{\alpha\beta}(x) \dot x^\alpha \dot x^\beta.
\label{L}
\end{equation}
The equation of motion (\ref{geo}) and all its differential consequences defines a submanifold ${\cal E} \subset {\cal J}$ called the {\it prolonged equation manifold}.  In the coordinates $(s, x^\alpha, \dot x^\alpha, \ddot x^\alpha, \dots)$ on $\cal J$ the equations for $\cal E$ are
\begin{equation}
\ddot x^\alpha = - \Gamma^\alpha_{\beta\gamma}(x) \dot x^\beta \dot x^\gamma ,\quad \dddot x^\alpha = D\left(- \Gamma^\alpha_{\beta\gamma}(x) \dot x^\beta \dot x^\gamma\right),\quad \cdots \ .
\label{eqmanifold}
\end{equation}

We now define a conservation law as a quantity built from any parametrized curve $x^\alpha = u^\alpha(s)$ and its derivatives to any order which becomes independent of $s$ when the curve satisfies the geodesic equation (\ref{geo}).  (The conserved quantity may depend explicitly upon $s$.) This means that the conservation law is a function on $\cal J$ whose total derivative vanishes when evalutated on the prolonged equation manifold. 
\begin{definition}
A function $Q\colon {\cal J}\to {\bf R}$ defines a {\it conservation law}  for a system of differential equations with prolonged equation manifold $\cal E$ if 
\begin{equation}
\left(DQ\right)\Big|_{\cal E} = 0.
\end{equation}
\end{definition}

A familiar example of a conservation law for the affinely parametrized geodesic equation is provided by the Lagrangian itself:
\begin{equation}
L = -\frac{1}{2}g_{\alpha\beta}(x) \dot x^\alpha \dot x^\beta.
\label{energy}
\end{equation}
We have
\begin{equation}
DL = -\frac{1}{2} g_{\alpha\beta,\gamma} \dot x^\alpha \dot x^\beta \dot x^\gamma - g_{\alpha\beta}\dot x^\alpha \ddot x^\beta = -g_{\alpha\delta}\Gamma^\delta_{\beta\gamma} \dot x^\alpha\dot x^\beta \dot x^\gamma -  g_{\alpha\beta}\dot x^\alpha \ddot x^\beta,
\end{equation}
which vanishes when evaluated on $\cal E$ as defined in (\ref{eqmanifold}).

Noether's theorem establishes a correspondence between conservation laws and symmetries of the Lagrangian provided the notion of ``symmetry'' is as follows. 
\begin{definition}
A generalized vector field (\ref{gvf}) defines a {\it generalized symmetry of a Lagrangian} $L$ if the infinitesimal transformation it defines leaves $L$ unchanged up to a total derivative, that is, there exists a function $G\colon \J\to {\bf R}$ such that
\begin{equation}
{\rm pr}\, {\bf v}(L[x])  + L[x] D(a[x]) = D G[x].
\label{gensymL}
\end{equation}
\end{definition}
Notice that the left hand side of (\ref{gensymL}) is just the Lie derivative of $L$ along ${\rm pr}\, {\bf v}$, taking account of the fact that the Lagrangian is a density of weight one on the real line with coordinate $s$, or equivalently that $L\,ds$ is a 1-form.
A straightforward calculation establishes the following convenient result \cite{Olver}.
\begin{proposition}  A generalized vector field defines a generalized symmetry of a Lagrangian if and only if its evolutionary representative does. 
\end{proposition}
\noindent For evolutionary vector fields the condition for symmetry of the Lagrangian is (\ref{gensymL}) with $a[x] = 0$. 

The connection between symmetries and conservation laws,  first proved by Noether, when specialized to first-order Lagrangians for ODEs is as follows \cite{Noether}, \cite{BesselHagen}, \cite{Prince2}, \cite{Olver}.

\begin{theorem} (Noether's Theorem)
\label{Noether}
The function $Q\colon \J \to {\bf R}$ defines a conservation law for the Euler-Lagrange equations of $L$ if and only if there exists a generalized vector field $\bf v$ defining a generalized symmetry (\ref{gensymL})  of $L$.   With ${\bf v}_{\rm ev} = \sigma^\alpha[x]{\partial\over\partial x^\alpha}$ and $L=L(s, x, \dot x)$, the conservation law is given by
\begin{equation}
Q[x] = {\partial L\over\partial \dot x^\alpha} \sigma^\alpha - G.
\end{equation}
\end{theorem}
\noindent For a general version of the theorem and proof, applicable to a general Lagrangian, and suitable for  ODEs or PDEs, see \cite{Olver} (see also \cite{Prince2} in the ODE context).

The conservation law defined by (\ref{energy}) can be obtained from Noether's theorem by virtue of the point symmetry 
\begin{equation}
{\bf v} = {\partial\over \partial s}\quad  \Longrightarrow\quad  {\bf v}_{\rm ev} = - \dot x^\alpha{\partial\over\partial x^\alpha}.
\end{equation}  
It is a classical result that spacetime isometries define symmetries and conservation laws for geodesics. This can be seen as follows.  A 1-parameter family of isometries $\phi_\lambda\colon M\to M$, $\phi_\lambda^*g = g$, $\lambda\in {\bf R}$, of a spacetime $(M, g)$ is generated by a Killing vector field $\xi$ on $M$ satisfying $L_\xi g = 0$.  This vector field lifts to an evolutionary vector field,
\begin{equation}
{\bf v} = \xi^\alpha(x) {\partial\over \partial x^\alpha},
\label{isolift}
\end{equation}
generating a point transformation of $\cal C$ and corresponding to the infinitesimal transformation of curves $x^\alpha = u^\alpha(s)$ given by
\begin{equation}
\delta u^\alpha(s) = \xi^\alpha(u(s)).
\end{equation}
The vector field (\ref{isolift}) defines a (point) symmetry of the Lagrangian (\ref{L}) for affinely parametrized geodesics:
\begin{equation}
{\rm pr}\, {\bf v} (L) = -\half (L_\xi g_{\alpha\beta}) \dot x^\alpha\dot x^\beta = 0.
\end{equation} 
From Theorem \ref{Noether} the conservation law is 
\begin{equation}
Q = - g_{\alpha\beta}\dot x^\alpha \xi^\alpha.
\end{equation}
It is straightforward to verify directly that $DQ = 0$ on the prolonged equation manifold $\cal E$ (see (\ref{eqmanifold})). 

\subsection{\it Symmetries and conservation laws associated to affine vector fields}

We now turn to one of the principal results of this paper: affine motions define  generalized  symmetries of the Lagrangian for affinely parametrized geodesics.  

\begin{theorem}
\label{symthm}
Let $Y$ be an affine vector field on the pseudo-Riemannian manifold $(M, g)$. Define $h = L_Y g$.  The generalized vector fields
\begin{align}
{\bf v}_1 &=  h^\alpha_\beta(x) \dot x^\beta {\partial\over \partial x^\alpha}\\
{\bf v}_2 &= \left(Y^\alpha(x) - s h^\alpha_\beta(x) \dot x^\beta\right){\partial\over \partial x^\alpha}
\end{align}
define generalized symmetries of the Lagrangian (\ref{L}).
\end{theorem}

\begin{proof}
The proof goes by direct computation. We use
\begin{equation}
g_{\alpha\beta, \gamma} = g_{\alpha\delta}\Gamma^\delta_{\beta\gamma} + g_{\beta\delta}\Gamma^\delta_{\alpha\gamma},
\end{equation}
\begin{equation}
h_{\alpha\beta, \mu} - \Gamma^\nu_{\mu\alpha} h_{\nu\beta} - \Gamma^\nu_{\beta\mu} h_{\alpha\nu} = 0,
\end{equation}
\begin{equation}
{\rm pr}\, {\bf v}_1(x^\alpha) = h^\alpha_\beta \dot x^\beta,
\end{equation}
\begin{equation}
{\rm pr}\, {\bf v}_1(\dot x^\beta) = h^\beta_{\sigma, \mu} \dot x^\mu\dot x^\sigma + h^\beta_\sigma \ddot x^\sigma,
\end{equation}
\begin{equation}
{\rm pr}\, {\bf v}_1(L) = -h_{\alpha\sigma}\dot x^\alpha\ddot x^\sigma  - h_{\alpha\beta, \mu}\dot x^\alpha \dot x^\beta \dot x^\mu   +\dot x^\alpha\dot x^\mu \dot x^\sigma h^\beta_\sigma g_{\alpha\beta, \mu} -\frac{1}{2} g_{\alpha\beta, \gamma} h^\gamma_\sigma \dot x^\sigma \dot x^\alpha \dot x^\beta
\end{equation}
\begin{equation}
{\rm pr}\, {\bf v}_2(x^\alpha) = Y^\alpha - s  h^\alpha_\beta \dot x^\beta,
\end{equation}
\begin{equation}
{\rm pr}\, {\bf v}_2(\dot x^\alpha) = Y^\alpha_{,\gamma}\dot x^\gamma -  h^\alpha_\beta \dot x^\beta - s D ( h^\alpha_\beta \dot x^\beta)
\end{equation}
\begin{align}
{\rm pr}\, {\bf v}_2(L) = &-g_{\alpha\beta} \dot x^\alpha ( Y^\beta_{,\gamma}\dot x^\gamma - h^\beta_\gamma \dot x^\gamma -sh^\beta_\sigma \ddot x^\sigma -s h^\beta_{\sigma, \mu} \dot x^\mu\dot x^\sigma )\nonumber\\ 
&- \frac{1}{2} g_{\alpha\beta, \gamma} (Y^\gamma - sh^\gamma_\sigma \dot x^\sigma) \dot x^\alpha \dot x^\beta
\end{align}

Combining all these relations reveals in each case
\begin{align}
{\rm pr}\, {\bf v}_1(L) &= D (-\frac{1}{2}h_{\alpha\beta}(x) \dot x^\alpha \dot x^\beta),\\ 
{\rm pr}\, {\bf v}_2(L)&= D(\frac{1}{2} s h_{\alpha\beta} \dot u^\alpha \dot u^\beta).
\end{align}

$\square$
\end{proof}

The existence of the two conservation laws for the geodesic equation  associated to an affine motion  \cite{Katzin}, \cite{Hojman1} now follows from an  application of Noether's theorem.

\begin{corollary}
\label{cor}
Let $Y$ be an affine vector field on the pseudo-Riemannian manifold $(M, g)$. Define $h = L_Y g$. The following quantities are conservation laws for the affinely parametrized geodesic equation:
\begin{align}
Q_1 &= \frac{1}{2} h_{\alpha\beta}(x) \dot x^\alpha \dot x^\beta,\\
 Q_2 &= Y_\alpha(x) \dot x^\alpha - s Q_1.
\end{align}
\end{corollary}

\begin{proof}
This follows from Theorem \ref{symthm}  and Theorem \ref{Noether} applied to ${\bf v}_1$ and ${\bf v}_2$ with the function $G$ given by  $-\frac{1}{2}h_{\alpha\beta}(x) \dot x^\alpha \dot x^\beta$ and $\frac{1}{2} s h_{\alpha\beta} \dot x^\alpha \dot x^\beta$, respectively.  It can also be verified directly by computing $D Q$ in each case and checking that this derivative vanishes on the prolonged equation manifold $\cal E$ by virtue of $\nabla h = 0$. $\square$
\end{proof}

Theorem \ref{symthm} proves that affine motions define generalized symmetries of the Lagrangian for affinely parametrized geodesics.  In the following we shall show in some detail how this occurs via the classification of affine symmetries due to Aminova \cite{Aminova}. 

Recall that an affine vector field either generates a 1-parameter group of homotheties or a 1-parameter group of proper affine symmetries.  The proper affine motions have been classified in \cite{Aminova} and define  {\it bona fide} generalized symmetries via Theorem \ref{symthm}.  In the case of homotheties, the vector fields in Theorem \ref{symthm} reduce to evolutionary representatives of point symmetries because $h = g$ in this case.    Let us begin by explicitly describing the point transformations associated to homothetic vector fields.  If  the affine vector $Y$ field generates a 1-parameter group of homotheties $\phi_t\colon M \to M$, $t\in {\bf R}$,
 \begin{equation}
 \phi_t^* g = e^t g,\quad L_Y g = g,
 \label{homo}
 \end{equation} 
the corresponding point transformation $\Phi^{(1)}_t\colon \C \to \C$ generated by ${\bf v}_1 = \dot x^\alpha {\partial\over \partial x^\alpha}$ in Theorem \ref{symthm} is given by a translation of $s$:
\begin{equation}
\Phi^{(1)}_t(s, x^\alpha) = (s - t, x^\alpha).
\end{equation}
Of course, the Lagrangian 1-form $L\, ds$ with $L$ given in (\ref{L}) has manifest symmetry under translation in $s$. The point transformation $\Phi^{(2)}_t\colon \C \to \C$ generated by ${\bf v}_2 = (Y^\alpha - s\dot x^\alpha){\partial\over\partial x^\alpha}$ is given by
\begin{equation}
\Phi^{(2)}_t(s, x^\alpha) = (e^{-t} s, \phi_t^\alpha(x)),
\label{homo2}
\end{equation}
corresponding to the homothetic mapping of $M$ onto itself along with a rescaling of the affine parameter.  This  transformation is also  a symmetry of $L\, ds$:  the homothety has the effect of rescaling the metric in (\ref{L}) which is compensated by the rescaling of $s$. In the special case where $Y$ generates an isometry, ${\bf v}_1$ vanishes while ${\bf v}_2$ reduces to  the  lift  (\ref{isolift}) of the infinitesimal generator of the isometry to $\cal C$.

We now explain in  detail how proper affine motions of Lorentz manifolds yield the generalized symmetries displayed in Theorem \ref{symthm}. According to reference \cite{Aminova}, if $Y$ is a proper affine vector field then (I) it acts  by homothety on irreducible components in a de Rham-Wu decomposition (\ref{prod}) and/or (II) it acts via (\ref{irr}) on  an irreducible Lorentzian component, and/or (III) it acts by the  intermixing transformations generated by vector fields of the form (\ref{mix}).   We now examine each of these cases. 

In  case (I), in coordinates adapted to the product, the Lagrangian decomposes according to (\ref{prod}),
\begin{equation}
L =- \frac{1}{2} g(\dot x, \dot x) =  -\frac{1}{2} g_0(\dot x_0, \dot x_0) - \frac{1}{2} g_1(\dot x_1, \dot x_1) - \dots - \frac{1}{2} g_r(\dot x_r, \dot u_r),
\label{Lcase1}
\end{equation}
and the action of the symmetry is to rescale each of the metrics $g_i$, $i = 1,\dots, r,$ with a corresponding rescaling of the affine parameter, as discussed above when a homothety of the entire metric was considered. See (\ref{homo2}).   The action of the symmetry on the first term in (\ref{Lcase1}) is as follows.  Introduce coordinates which rectify a basis of parallel vector fields so that the metric takes the form
\begin{equation}
g_0 = \epsilon dx \otimes dx + dy \otimes dy + \cdots,
\end{equation}
where $\epsilon = \pm1$.  The proper affine vector fields can be chosen to take the form
\begin{equation}
Y = y \partial_x.
\end{equation}
We then have
\begin{align}
{\bf v}_1 &= \dot y {\partial\over\partial x} + \epsilon \dot x{\partial\over\partial y}\\
{\bf v}_2 &= (y - \epsilon s \dot y){\partial\over\partial x} + \epsilon s \dot x{\partial\over\partial y}.
\end{align} 
The corresponding portion of the Lagrangian transforms as
\begin{align}
{\rm pr}\, {\bf v}_1 \left(-\half g_0(\dot x, \dot x)\right)   &= D (-\epsilon\dot x \dot y),\\
{\rm pr}\, {\bf v}_2\left(-\half g_0(\dot x, \dot x)\right)  &= D(\epsilon s \dot x \dot y).
\end{align}

%
In case (II) it is convenient to use the coordinates $x^\alpha = (v, w, x^i)$ introduced in (\ref{irr2}), (\ref{irr3}).
The Lagrangian (\ref{L})  is given in such coordinates by
\begin{equation}
L = - \dot v \dot w - \frac{1}{2} g_{ij}(w, x^k) \dot x^i \dot x^j.
\end{equation}
The infinitesimal symmetry transformation is generated by 
\begin{align}
{\bf v}_1 &= 2\dot w {\partial\over\partial v} \\
{\bf v}_2 &= (w - 2s \dot w){\partial\over\partial v} 
\end{align}
yielding the symmetries of $L$
\begin{equation}
{\rm pr}\, {\bf v}_1(L) = D (-   (\dot w)^2), \quad {\rm pr}\, {\bf v}_2(L) = D\left[ s (\dot w)^2\right].
\end{equation}

Finally, in case (III), aside from irreducible components and their affine symmetries of type (I),  we have a metric of the form
\begin{equation}
g = \delta_{AB} d\rho^A\otimes d\rho^B + dw \otimes dv + dv\otimes dw + g_{ij}(w, z) dz^i \otimes dz^j.
\end{equation}
The intermixing transformations are of the form
\begin{equation}
Y =   w\partial_{\rho^A}.
\end{equation}
The infinitesimal transformations with, for example, $A=1$ are given by
\begin{align}
{\bf v}_1 &= \dot w {\partial\over \partial x^1} + \dot \rho^1 {\partial\over\partial v}\\
{\bf v}_2 &= (w - s\dot w){\partial\over\partial x^1} - s\dot \rho^1 {\partial\over\partial v}.
\end{align}
The Lagrangian
\begin{equation}
L  = -\half\left(\delta_{AB} \dot \rho^A \dot \rho^B + 2 \dot w \dot v + g_{ij}(w, z) \dot z^i \dot z^j\right)
\end{equation}
transforms as
\begin{equation}
{\rm pr}\, {\bf v}_1(L)  = D(- \dot \rho^1 \dot w),\quad {\rm pr}\, {\bf v}_2(L) = D(s \dot \rho^1 \dot w).
\end{equation}

\section{Affine Motions and Conservation Laws for Homogeneous Solutions of the Einstein Equations}
\label{solsection}

A number of authors have found  affine motions for various solutions of the Einstein equations, {\it e.g.,} \cite{Beldran}, \cite{Maartens}, \cite{Collinson}, \cite{Hojman2}. Besides the classification of affine motions due to Aminova \cite{Aminova}, Hall and da Costa used a classification of holonomy groups to determine which  affine symmetries may occur in a four-dimensional spacetime \cite{Hall}.   The possibilities for electrovacua in four dimensions have been examined in \cite{Carot}.  
 In this section we calculate all proper affine motions which arise for homogeneous solutions of the Einstein field equations.  This provides a complete characterization  of proper affine motions for all homogeneous solutions of the Einstein equations with matter content given by vacuum, Einstein, perfect fluid, and electromagnetic field. We also give the corresponding conservation laws for affinely parametrized geodesics.

Homogeneous spacetimes admit a transitive group of isometries.  All homogeneous solutions to the Einstein equations,
\begin{equation}
G_{ab} + \Lambda g_{ab} = \kappa T_{ab},
\end{equation}
are known for the following cases: vacuum ($T_{ab}=0 = \Lambda$);  Einstein ($T_{ab}=0$, $\Lambda \neq 0$); homogeneous perfect fluids,
\begin{equation}
T_{ab} = (\mu + p)V_a V_b + p g_{ab},\quad \mu, p = {\rm const.},\quad V_a V^a = -1,\quad  \Lambda = 0;
\end{equation}
 pure radiation,
\begin{equation}
T_{ab} =\Phi^2 k_a k_b,\quad \Phi = const., \quad \Lambda = 0, \quad k_ak^a = 0;
\end{equation}
and  homogeneous electromagnetic fields,
\begin{equation}
T_{ab} = F_a{}^c F_{bc} - \frac{1}{4} g_{ab} F_{cd} F^{cd},\quad \nabla^a F_{ab} = 0 = \nabla_{[a}F_{bc]}, \quad \Lambda = 0.
\end{equation}
The 4-velocity, radiation vector, and electromagnetic field satisfy  $ L_\xi V = 0$, $L_\xi k = 0$,  $L_\xi F = 0$, where $\xi$ is any Killing vector field, $L_\xi g = 0$.  All these solutions can be found in reference \cite{Stephani}.

Using the {\sl DifferentialGeometry} package \cite{DG} we have calculated all the  affine vector fields and first integrals for this class of solutions.  The analysis has two parts.  First, one directly solves the equations arising from infinitesimal invariance of the Christoffel symbols in the given coordinate chart:
\begin{equation}
L_Y\Gamma_{\alpha\beta}^\gamma = 0.
\label{ACeq}
\end{equation}
It is straightforward to extract from the solution space of (\ref{ACeq}) a basis for the set of proper affine vector fields (modulo homotheties). Second, the results are checked against the dimension of the vector space of solutions to (\ref{ACeq}), which can be computed {\it a priori} as follows.   

The linear system of equations (\ref{infAffine}) determining $Y$ is of {\it finite type}, so that all second and higher order derivatives of $Y$ are determined by $Y$ and its first derivatives. The dimension of the vector space of solutions in a neighborhood of a   point $p\in M$ can be determined by successively differentiating equation (\ref{infAffine}) and expressing the result in terms of $Y$ and $\nabla Y$ at $p$. This defines a system of linear equations for the $n(n+1)$ dimensional vector space of data $Y(p)$ and $\nabla Y(p)$.  If $r$ is the rank of this linear system, then the vector space of solutions to (\ref{infAffine}) is of dimension $n(n+1) - r$. 

In the following, we list the results of this analysis for all homogeneous solutions of the Einstein equations with matter  content as described above. We follow the enumeration of these solutions as given in reference \cite{Stephani}.  We present the line element of the metric, a basis for the vector space of proper affine vector fields (modulo homothetic vector fields), and the corresponding conservation laws for geodesics, calculated according to Corollary \ref{cor}. 
\begin{description}

%
%
%
 
 \item
 
  \item[\bf  Solution:] Minkowski spacetime. Coordinates $x^\alpha = (t, x, y, z)$.
 
 \item[\bf  Line element:] $ds^2 = -dt^2 + dx^2 + dy^2 + dz^2$
 
  \item[\bf  Proper affine vector fields:]  $x\partial_x$, $y\partial_y$,  $z\partial_z$, $ x^k\partial_t$, $x^i\partial_j$,  $i > j$,  $x^i = (x, y, z)$ .
 
 \item[\bf  First integrals:] $(\dot x^i)^2$, $x^i \dot x^i - s  (\dot x^i)^2$;  $ \dot t \dot x^i$, $x^i \dot t - s\dot t\dot x^i$; $\dot x^i\dot x^j$, $ x^i \dot x^j - s\dot x^i \dot x^j$, $i>j$
 
 \item

  \item[\bf  Solution:] de Sitter, and anti-de Sitter spacetime; eq. (8.33) in \cite{Stephani}. Coordinates $x^\alpha = (t, x, y, z)$.
 
 \item[\bf  Line element:] $ds^2 = \Psi^{-2}\left(-dt^2 + dx^2 + dy^2 + dz^2\right)$; $\Psi = 1 + {1\over 4} K(x^2 + y^2 + z^2 - t^2)$, $K\neq0$.
 
  \item[\bf  Proper affine vector fields:] $\nexists$
 
 \item[\bf  First integrals:] $\nexists$
 
  \item
\item[\bf Solution:] Petrov vacuum solution; eq. (12.14) in \cite{Stephani}. Coordinates $x^\alpha = (t, x, y, z)$.

 \item[\bf Line element:]  $k^2 ds^2 = dx^2 + e^{-2x}dy^2 + e^x [\cos\sqrt{3}x(dz^2 - dt^2) - 2\sin\sqrt{3}x \, dz \, dt]$

 \item[\bf Proper affine vector field:] $\nexists$

 \item[\bf First integrals: ] $\nexists$

 \item

%
%
%
%
 
 \item
 
  \item[\bf  Solution:] 
Einstein spaces  arising as a product of 2-dimensional spaces of constant curvature, $S^2 \times S^2$ and $H^2 \times H^2$;   eq. (12.8) in \cite{Stephani}. Coordinates $x^\alpha = (t, x, y, z)$.

 \item[\bf  Line element:] 
$
 ds^2 = A^2[dx^2 + \Sigma^2(x,k_1)dy^2] + B^2[dz^2 - \Sigma^2(z,k_2)dt^2],
$ where $k_1 = k_2 = \pm1,\quad 1/A^2 = 1/B^2 = \Lambda$, 
$ \Sigma(x, 1) = 
 \sin(x), \quad \Sigma(x, -1) = 
 \sinh(x),
 $
 \item[\bf  Proper affine vector field:] $\nexists$
 
 \item[\bf First Integrals:] $\nexists$
 
 \item

 \item[\bf Solution:]  Einstein space of Petrov type N; equation (12.34) in \cite{Stephani}. Coordinates $x^\alpha = (u, x, y, z)$.
 
 \item[\bf Line element:] $ds^2 = 3dz^2/|\Lambda| + \varepsilon e^z dx^2 + e^{-2z}(dy^2 + 2 du \, dx)$, $\Lambda < 0$, $\varepsilon = \pm1$.
 
  \item[\bf Proper affine vector fields:] $\nexists$
 
 \item[\bf First integrals:] $\nexists$
 
 \item

  \item[\bf Solution:]  Einstein space  of Petrov type III; equation (12.35) in \cite{Stephani}. Coordinates $x^\alpha = (u, x, y, z)$.
 
 \item[\bf Line element:] $ds^2 = 3 dz^2/|\Lambda| + e^{4z}dx^2 + 4e^z dx \, dy + 2e^{-2z}(dy^2 + du \, dx)$, $\Lambda<0$.
 
  \item[\bf Proper affine vector fields:] $\nexists$
 
 \item[\bf First integrals:] $\nexists$
 
 \item

  \item[\bf Solution:] Plane wave electrovacuum; eq. (12.12) with $\epsilon = 0$ in \cite{Stephani}. Coordinates $x^\alpha = (u, v, \xi, \overline\xi)$.

 \item[\bf Line element:]  $ds^2 = -\left[2a (\xi^{2} e^{-2i\gamma u}  + \overline\xi^2 e^{2i\gamma u}) + 2b^2 \xi\overline\xi\right] du^2 - 2du\,  dv  + 2d\xi \, d\overline \xi $

 \item[\bf Proper affine vector field:] $Y = u\, \partial_v$

 \item[\bf First integrals: ] $\dot u^2$, $u\dot u - s \dot u^2$
   \item

  \item[\bf Solution:] Plane wave electrovacuum; eq. (12.12) with $\epsilon = 1$ in \cite{Stephani}. Coordinates $x^\alpha = (u, v, \xi, \overline\xi)$.

 \item[\bf Line element:]  $ds^2 = -\left[2a (\xi^{2} e^{-2i\gamma u}  + \overline\xi^2 e^{2i\gamma u}) + 2b^2 \xi\overline\xi\right] du^2 - 2e^u du\,  dv  + 2d\xi \, d\overline \xi $

 \item[\bf Proper affine vector field:] $Y = e^u\, \partial_v$

 \item[\bf First integrals: ] $e^{2u}\dot u^2$, $e^{2u}(\dot u - s \dot u^2)$
 
 \item

 \item

  \item[\bf Solution:] Bertotti-Robinson electrovacuum; eq. (12.16) in \cite{Stephani}. Coordinates $x^\alpha = (t, x, \theta, \phi)$.

 \item[\bf Line element:]  $ds^2 = k^2\left[d\theta^2 + \sin^2(\theta) d\phi^2 + dx^2 - \sinh^2(x) dt^2\right]$

 \item[\bf Proper affine vector field:] $\nexists$

 \item[\bf First integrals: ] $\nexists$

\item
 
 \item[\bf Solution:] Pure radiation\footnote{This pure radiation solution is actually an electrovacuum with an inhomogeneous electromagnetic field \cite{Torre}. }; equation (12.36) in \cite{Stephani}. Coordinates $x^\alpha = (u, v, x, y)$.
 
 \item[\bf Line element:] $ds^2 = dx^2 + dy^2 + 2 du dv -2 e^{2\rho x} du^2$, $\rho = const.\neq 0$
 
 \item[\bf Proper affine vector fields:] $y\partial_y$, $u\partial_v$, $u\partial_y$
 
 \item[\bf First integrals:] $\dot y^2$, $y\dot y - s \dot y^2$, $\dot u^2$, $u\dot u - s\dot u^2$, $\dot u \dot y$, $u\dot y -s\dot u\dot y$
  
\end{description}

 \item

 \item[\bf Solution:] Einstein static universe (and its hyperbolic analog). Perfect fluid source; eq. (12.9) and (12.24) in \cite{Stephani}.  Coordinates $(t, r, \theta, \phi)$.
  
 \item[\bf Line element:] 
 \begin{equation}
 ds^2 = a^2 \left[dr^2 + \Sigma^2(r, k)(d\theta^2 + \sin^2\theta d\phi^2)\right] - dt^2,\quad k = \pm 1.
 \end{equation}
 
 \item[\bf Proper affine vector fields:] $Y= t\, \partial_t$
 
 \item[\bf First integrals:] $ \dot t^2$, $t\dot t - s \dot t^2$
  
 \item

 \item[\bf Solution:] G\"odel spacetime.  Perfect fluid source; eq. (12.26) in \cite{Stephani}. Coordinates $x^\alpha = (t, x, y, z)$.
 
 \item[\bf Line element:] $ds^2 = a^2[dx^2 + dy^2 + \frac{1}{2}e^{2x}dz^2 - (dt+ e^{x} dz)^2]$
 
 \item[\bf Proper affine vector fields:] $Y = y\, \partial_y$
 
 \item[\bf First integrals:] $\dot y^2$, $y\dot y - s \dot y^2 $
 
 \item
 
 \item[\bf Solution:] Unimodular group. Farnsworth-Kerr class I. Perfect fluid source; eq. (12.27) in \cite{Stephani}. Coordinates $x^\alpha = (t, x, y, z)$.
 
 \item[\bf Line element:] $ds^2 = a^2[(1-k)( {\omega^1})^2 + (1+k)( {\omega^2})^2 + 2( {\omega^3})^2 - (dt + \sqrt{1 - 2k^2} {\omega^3})^2]$, $d\omega^a = \epsilon^a{}_{bc} \omega^b\wedge \omega^c$, $0<|k|<\frac{1}{2}$
 
  \item[\bf Proper affine vector fields:] $\nexists$
 
 \item[\bf First integrals:] $\nexists$

\item

\item[\bf Solution:] Unimodular group. Farnsworth-Kerr class II. Perfect fluid source; eq. (12.28) in \cite{Stephani}. 
 
 \item[\bf Line element:] $ds^2 = a^2[(1-k)( {\omega^1})^2 + (1+k)( {\omega^2})^2 + (du + \sqrt{1-2k^2}  {\omega^3})^2 - 2( {\omega^3})^2]$, $d\omega^1 = \omega^2\wedge\omega^3, d\omega^2 = \omega^3 \wedge\omega^2, d\omega^3 = \omega^2\wedge\omega^1$,  $1 < 4k^2 < 2$,
 
  \item[\bf Proper affine vector fields:] $\nexists$
 
 \item[\bf First integrals:] $\nexists$

 \item
 
\item[\bf Solution:] Unimodular group. Farnsworth-Kerr class III. Perfect fluid source; eq. (12.29) in \cite{Stephani}. 
 
 \item[\bf Line element:] $ds^2 = a^2[(1-s)( {\omega^1})^2 + (1 + s)( {\omega^2})^2 + du^2 - 2( {\omega^3})^2]$, $d\omega^1 = \omega^2 \wedge \omega^3$, $d\omega^2 = \omega^3\wedge \omega^1$, $d\omega^3 = \omega^2\wedge \omega^1$, $|s|<1$.
 
 \item[\bf Proper affine vector fields:] $Y = u\partial_u$
 
 \item[\bf First integrals:]  $\dot u^2$, $u \dot u - s \dot u^2$.
 
 \item

  \item[\bf Solution:] Non-unimodular group.  Perfect fluid source; eq. (12.30) in \cite{Stephani} with generic value for parameter $s$. Coordinates $x^\alpha = (t, x, y, z)$.
 
 \item[\bf Line element:] $ds^2 = a^2[4(Ae^{Az}dt +Be^{Bz}dx)^2/b^2 + (e^{Fz}dy)^2 + dz^2 - (e^{Az}dt + e^{Bz} dx)^2]$, $A = \frac{1}{2}(1-\beta)$, $B = \frac{1}{2}(1+\beta)$, $F=1-s^2$, $b=\sqrt{2} s(2-s^2)$, $\beta^2 = 1+2s^2(1-s^2)(3-s^2) >0$, $\frac{1}{2}\leq s^2\leq 2$.
 
  \item[\bf Proper affine vector fields:] $\nexists$
 
 \item[\bf First integrals:] $\nexists$

 \item

  \item[\bf Solution:] Non-unimodular group.  Perfect fluid source; eq. (12.31) in \cite{Stephani} with generic value for parameter $s$. Coordinates $x^\alpha = (t, x, y, z)$.
 
 \item[\bf Line element:] $ds^2  = a^2[e^z dx^2 + e^{2Fz}dy^2 + dz^2 - \frac{1}{4}(b-1/b)^2 e^{z}(dt+z \, dx)^2]$,  $b=\sqrt{2} s(2-s^2)$, $F=1-s^2$, $\frac{1}{2}\leq s^2\leq 2$
 
  \item[\bf Proper affine vector fields:] $\nexists$
 
 \item[\bf First integrals:] $\nexists$

 \item

  \item[\bf Solution:] Non-unimodular group.  Perfect fluid source; eq. (12.32) in \cite{Stephani} with generic value for parameter $s$. Coordinates $x^\alpha = (t, x, y, z)$.
 
 \item[\bf Line element:] 
 \begin{align}
 ds^2 &= a^2\Big[(e^{Fz} dy)^2 + dz^2 + e^z ((\cos \,kz - 2k \sin \,kz)dt \nonumber \\
 &+ (2k \cos \,kz + \sin \,kz) dx)^2/b^2- e^z(\cos \,kz \, dt + \sin \,kz \, dx)^2\Big],  \nonumber\\
  &b=\sqrt{2} s(2-s^2), \quad F=1-s^2,\quad \frac{1}{2}\leq s^2\leq 2\nonumber
\end{align} 
 
  \item[\bf Proper affine vector fields:] $\nexists$
 
 \item[\bf First integrals:] $\nexists$

\item

\vfill\eject

\section*{Acknowledgment}

This work was supported in part by National Science Foundation grant ACI-1642404 to Utah State University (CT) and a USU-Physics Howard Blood Fellowship (DM).

\end{document}